\documentclass[preprint, 10pt]{elsarticle}
\usepackage{graphicx}
 \expandafter\let\csname equation*\endcsname\relax
 \expandafter\let\csname endequation*\endcsname\relax
\usepackage{amsmath}
\usepackage{amsthm}
\usepackage{subfig}
\usepackage{bm}
\usepackage{xcolor}
\usepackage{booktabs}

\newtheorem{theorem}{Theorem}
\newtheorem{example}{Example}

\newcommand{\eps}{\varepsilon}
\renewcommand{\d}{\mathrm{d}}

\renewcommand{\vec}[1]{{\mathbf{#1}}}

\newcommand{\partialx}{\frac{\partial}{\partial x}}
\newcommand{\partialy}{\frac{\partial}{\partial y}}
\newcommand{\partialzeta}{\frac{\partial}{\partial\zeta}}
\newcommand{\partialeta}{\frac{\partial}{\partial \eta}}

\begin{document}

\title{Streamline integration as a method for structured grid generation in X-point geometry }
\author[1]{M. Wiesenberger}
\ead{mattwi@fysik.dtu.dk}
\author[2]{M. Held}
\author[3,4]{L. Einkemmer}
\author[2]{A. Kendl}
\address[1]{Department of Physics, Technical University of Denmark (DTU), 2800 Kgs. Lyngby, Denmark}
\address[2]{Institute for Ion Physics and Applied Physics,
  Universit{\"a}t Innsbruck, 6020 Innsbruck, Austria}
\address[3]{Department of Mathematics, Universit{\"a}t T\"ubingen, 72076 T\"ubingen, Germany}
 \address[4]{Department of Mathematics, Universit\"at Innsbruck, 6020 Innsbruck, Austria}

\begin{abstract}
  We investigate structured grids aligned to the contours of a two-dimensional flux-function with an X-point (saddle point).
  Our theoretical analysis finds that orthogonal grids
exist if and only if the Laplacian of the flux-function vanishes
at the X-point.
  In general, this condition
  is sufficient for the existence of a structured aligned grid with an X-point.
  With the help of streamline integration we then
  propose a numerical grid construction algorithm.
  In a suitably chosen monitor metric the Laplacian of the flux-function vanishes
  at the X-point such that a grid construction is possible.

  We study the convergence
  of the solution to elliptic equations
  on the proposed grid.
  The diverging volume element and cell sizes at the X-point reduce the convergence
  rate. As a consequence, the proposed grid should
  be used with grid refinement around the
  X-point in practical applications. We show that grid refinement in the cells neighbouring the X-point
  restores the expected convergence rate.

\end{abstract}

\maketitle
\section{Introduction}
A magnetic X-point is particularly advantageous for the confinement of particles and thermal energy inside a magnetic fusion device~\cite{Wesson}.
For this reason, two- and three-dimensional simulations that encompass the X-point in the cross-section
of magnetically confined fusion plasmas have
emerged in past years~\cite{Rognlien1992,Schneider2006, Huysmans2007,Xu2008, Chang2009, Hoelzl2012,Dudson2015, Tamain2016, Dudson2017, Reiser2017, Galassi2017}.
There, so-called \textit{flux-surfaces}~\cite{Wesson} bound the idealized toroidally symmetric physical
domain.
Analytically, the flux-surfaces are represented by the contour lines of the \textit{flux-function} $\psi$, which at an X-point has a vanishing gradient and an indefinite Hessian matrix.
It has proven advantageous to use
grid points that align with this flux-function in numerical simulations.
This is especially true in the closed field line region, where flux-aligned structures
like zonal flows regulate the turbulent transport~\cite{Diamond2005}.
Furthermore, once the domain of interest is bounded by flux-surfaces,
a "flux-aligned" grid allows for an easy treatment of boundary conditions.

Unfortunately, structured grids (grids
generated by a coordinate transformation) aligned to flux-surfaces may lead to
numerical issues when an X-point is present in the domain.
This is because one coordinate of a structured aligned grid is necessarily the
flux-function $\psi$ itself or a monotonous function of it\footnote{
  This is just a re-expression of the alignment condition.}.
  Since $\psi$ has per definition
a saddle point with $\nabla\psi=0$, the Jacobian of the coordinate transformation
vanishes at this point and the transformation becomes singular.
This also entails vanishing or diverging elements in the metric tensor,
which appear in the physical equations transformed to the new coordinate system
and therefore enter the numerical discretization.
However, these issues do not directly manifest in the grid points
themselves.
In fact, with the help of streamline integration~\cite{haeseleer,Wiesenberger2017}
it is fairly straightforward to numerically construct grid points that
are aligned with the flux-surfaces. What is unclear is whether
\begin{itemize}
  \item these then
actually represent a (homeomorphic) coordinate transformation,
 \item a numerical
 scheme can cope with the singularity (consistence),
 \item the
convergence rate of a numerical scheme is affected by the singularity.
\end{itemize}
For example,
in an elliptic equation the solution depends
on all points in the domain and we cannot a priori know whether singular points reduce or prevent the
global convergence rate of the solution.
 We are the first to address these concerns, which have not been studied systematically in the literature so far.
Nevertheless, results from simulations on structured aligned grids have already been published~\cite{Tamain2016, Dudson2017, Reiser2017, Galassi2017} without
investigating or solving the above issues.
We believe that the therein presented conclusions require a discussion in the light of the numerical uncertainties and the results in the present article.

Let us mention that, of course, the use of coordinate patches or entirely unstructured grids is
always possible and circumvents the problem~\cite{Huysmans2007, Nishimura2006, Chang2009a,Zhang2016}.
Still, we investigate the use of structured grids in this contribution as they have
several advantages.
First of all, numerical methods on structured grids are very easily implemented.
Unstructured coordinate patches
introduce an overhead due to the additional bookkeeping induced by the explicit topological information
of grid patches or cells.
Furthermore, this overhead necessarily leads to a loss in performance over
structured grids since for
example in the computation of derivatives the additional topological information
needs to be separately loaded from the system memory.
This is detrimental for memory bandwidth bound problems.

For completeness let us also mention
recent approaches to use non flux-aligned grids for
the discretization of model equations~\cite{Hariri2014, Stegmeir2014, Held2016}.
Like unstructured grids, these avoid numerical issues with the X-point but shift the problem
to the question of how to correctly implement a flux-aligned boundary.

Finally, let us note that even though we motivated the problem from within the
field of magnetic confinement fusion, its nature is purely mathematical.
Our results therefore apply to any situation in which an alignment of a
numerical grid to a two-dimensional function with X-point is desirable.
Also note that in this contribution the discussion of \textit{O-points} (extrema of
the flux-function) is missing. This is because we assume the Hessian matrix of the
flux-function to be
indefinite in our derivation and the results therefore do not apply to O-points.

In this contribution we investigate how
structured grids can be consistently constructed and how
numerical methods behave
when there is an X-point present in the
computational domain.
In Section~\ref{sec:theory} we discuss general
properties of structured grids aligned to flux-surfaces from
an analytical point of view. We derive a consistency equation that all
structured grids aligned to a flux-function have to obey. Based on this
we derive necessary and sufficient conditions to fulfil this equation.
We then propose a grid generation algorithm for orthogonal grids in Section~\ref{sec:construction}.
Our algorithm is based on streamline integration~\cite{haeseleer, Wiesenberger2017} and assumes that the Laplacian of the flux-function vanishes at the X-point.
This technique allows the efficient computation of grid coordinates as well as the corresponding
Jacobian and therefore metric elements up to machine precision.
We pay special attention to the discretization of the separatrix (the contour line through the X-point).
In the following Section~\ref{sec:monitor} we then show how our algorithm applies to cases with
a non-vanishing Laplacian at the X-point. We introduce the concept of a \textit{monitor metric}.
Finally, in Section~\ref{sec:numerics} we apply our algorithm first to an
analytical example and second to a practical problem taken from the field of magnetically
confined fusion.
With the analytical example we in particular show how grid generation algorithms fail without monitor metric.
For the second case we solve an elliptic equation on our generated grid and
show convergence rates of a local discontinuous Galerkin discretization
of various order~\cite{Cockburn2001}.
If the solution varies across the X-point, we need grid refinement to restore the
convergence of our solution, which otherwise deteriorates to order one in the cell-size due to the
diverging volume element.


\section{Structured grids with X-point} \label{sec:theory}
Given is a two-dimensional
flux-function $\psi(x,y)$ in some coordinates $x$ and $y$. At one point $x_X,y_X$
this function has a saddle point (the {\it X-point}), where the gradient vanishes and the
Hessian matrix is indefinite.
Let us assume the existence of a metric tensor\footnote{This metric
later becomes the monitor metric.} $g$ with elements
given in the coordinates $x$ and $y$.
We now express a coordinate system $\zeta, \eta$ with $\zeta$ aligned to $\psi(x,y)$ as\footnote{
Here and in the following we use the notation
$\psi_x := \partial\psi/\partial x$, $\psi_{xx}:=\partial^2\psi/\partial x^2$, ... }
\begin{subequations}
\begin{align}
  \d \zeta & =  f(\psi)(\psi_x \d x + \psi_y \d y) \label{eq:orthogonala} \\
  \d \eta  & = a(x,y)\sqrt{g}[ -\psi^y \d x + \psi^x \d y ]  - b(x,y)[
  \psi_x \d x +  \psi_y \d y] \label{eq:orthogonalb}
\end{align}
  \label{eq:orthogonal}
\end{subequations}
Equation~\eqref{eq:orthogonala} expresses the alignment property $\d\zeta = f(\psi)\d\psi$
with $f(\psi) \neq 0$.
Our choice for the form of $\d\eta$ in Eq.~\eqref{eq:orthogonalb} becomes apparent further down in the text.
It is a re-expression of the general exact 1-form in two dimensions,
$\d \eta = \eta_x \d x + \eta_y \d y$.
In place of $\eta_x$ and $\eta_y$ we introduce the two free functions $a(x,y)\neq 0$ (if $f$ or $a$ were zero at a point, the coordinate transformation would become singular) and $b(x,y)$.
We have the contravariant components of $\nabla \psi$,
\begin{align*}
\psi^x := g^{xx}\psi_x + g^{xy}\psi_y,\quad
\psi^y := g^{xy}\psi_x + g^{yy}\psi_y
  \label{}
\end{align*}
and the element of the volume form
$\sqrt{g}:=(g_{xx}g_{yy}-g_{xy}g_{xy})^{1/2}$.
Then,
$\eta_x = - \sqrt{g}\psi^y a - \psi_x b$ and
$\eta_y =  \sqrt{g}\psi^x a - \psi_y b$,
which is invertible for $a$ and $b$ if
$(\nabla\psi)^2 = \psi^x\psi_x + \psi^y\psi_y \neq 0$.
In fact, we then have
\begin{align}
  a = \frac{\psi_x\eta_y-\psi_y\eta_x}{\sqrt{g}(\nabla\psi)^2}, \quad
  b = -\frac{\psi^x\eta_x + \psi^y \eta_y}{(\nabla\psi)^2}.
  \label{eq:ab}
\end{align}
Recall the familiar rules for tensor transformation (e.g.~\cite{Frankel}). The elements of the inverse metric tensor $g^{-1}$ in the transformed coordinates read
\begin{subequations}
\begin{align}
  \left .\begin{pmatrix}
    \bar g^{\zeta\zeta} & \bar g^{\zeta\eta} \\
    \bar g^{\eta\zeta} & \bar g^{\eta\eta}
  \end{pmatrix} \right|_{\zeta(x,y), \eta(x,y)}
   = \left(\nabla\psi\right)^2\left .\begin{pmatrix}
    f^2 & -bf \\
    -bf & a^2+b^2
  \end{pmatrix}\right|_{x,y}, \quad
\sqrt{ \bar g}^{\,-1} = (\nabla\psi)^2 a f,
\label{eq:metric_transformed}
\end{align}
\end{subequations}
which shows that we obtain an orthogonal grid (a grid in which the base vectors are
orthogonal in the given metric) with $b=0$. We denote $\bar g^{ij}$ as the elements of
$g^{-1}$ in the transformed coordinate system $\zeta,\eta$ and analogous $\sqrt{\bar g}$ the element
of the volume form in transformed coordinates.

  Let us emphasize here that the metric tensor $g$ is not necessarily the canonical, Cartesian metric.
  We only assume that $g$, as well as the coordinates $x$ and $y$, are well-defined and do not expose any singularities.
  Further note that our choice of notation is based on differential forms rather than what is
  traditionally used in the plasma physics literature~\cite{haeseleer}. In this way the relation between the metric tensor and the more fundamental objects (covariant and contravariant base vectors) is disentangled.
  This becomes advantageous in Sections~\ref{sec:construction} and \ref{sec:monitor}, where we want to freely choose the metric tensor.

Now, we place ourselves in a reverse position. If $\psi$ and the metric $g$ are given,
is it possible to find $\zeta$ and $\eta$ in
the form presented in Eq.~\eqref{eq:orthogonal}?
In fact, this question is equivalent to finding conditions for the functions $a$, $b$ and $f$
such that the right-hand sides of Eqs.~\eqref{eq:orthogonala} and \eqref{eq:orthogonalb} are exact forms.
Recall that the Poincar\'e lemma states that a closed form is exact~\cite{Frankel}.
Therefore, $f\d\psi$  has a potential $\zeta$ if $\d(f\d\psi)= 0$. This results in
\begin{align*}
  f_y \psi_x - f_x \psi_y = 0
  \label{}
\end{align*}
and
is fulfilled if $f$
is a function of $\psi$ only.
In order for the coordinate $\eta$ to exist it must hold that
$\d[(-a\sqrt{g}\psi^y-b\psi_x)\d x + (\sqrt{g}\psi^x a - \psi_y b)\d y] = 0$.
  In coordinates that is
\[ \partialx\left(\sqrt{ g}a\psi^x - b\psi_y\right) + \partialy\left(\sqrt{g} a\psi^y + b\psi_x\right) = 0. \]
This can be rewritten to
\begin{equation}
  a \Delta \psi + \nabla\psi\cdot \nabla a + \{ \psi, b\} = 0,
    \label{eq:consistency}
\end{equation}
where
$\Delta\psi$ is the Laplacian operator given by
\begin{align}
  \Delta\psi := \frac{1}{\sqrt{g}}\left(\partialx\left(\sqrt{g}\psi^x\right) + \partialy\left( \sqrt{g} \psi^y\right)\right)
  \label{eq:def_laplacian}
\end{align}
and we identified the Poisson bracket\footnote{
The interested reader will recognize that this is indeed the correct definition of the
Poisson bracket since the volume form in two-dimensions can be identified with
the symplectic (area) 2-form. The elements of the inverse symplectic form are the Poisson
brackets of the coordinates among themselves~\cite{Frankel}.}
\[\{\psi,b\} := \frac{1}{\sqrt{g}}(\psi_xb_y -\psi_y b_x).\]
It is the recovery of the Laplacian, the gradient and the Poisson bracket in Eq.~\eqref{eq:consistency} that justifies our choice of Eq.~\eqref{eq:orthogonalb}.

Since the flux-function $\psi(x,y)$ and the metric $g$ are given, Eq.~\eqref{eq:consistency} is a constraint on the functions
$a(x,y)$ and $b(x,y)$. We call Eq.~\eqref{eq:consistency}
the {\it consistency equation} and for the remainder of this section we focus on its
implications.
Apparently, the problematic point is the X-point,
where $\psi_x$ and $\psi_y$ vanish, but $\Delta\psi$ might not. We
therefore ask under what circumstances well-defined solutions $a$ and $b$ exist, depending on the
properties of $\psi$ at the X-point.
A vanishing Laplacian is in fact a very desirable quality of $\psi$.
At this point an example is instructive.
\begin{example} \label{ex:1} We consider $\psi = \tfrac{1}{2}(x^2-y^2)$ and $f(\psi)=1$ in the canonical (Cartesian) metric, for which $\Delta \psi = 0$.
  One possible choice for the second coordinate is $\eta = xy$. Equation \eqref{eq:ab} yields at once $a=1$ and $b=0$
    that is, we have obtained an orthogonal coordinate system ($b=0$).
    The non-zero metric elements are $\bar g^{\zeta\zeta}=\bar g^{\eta\eta}=x^{2}+y^{2}$ from which we can compute the volume element $\sqrt{\bar g} = 1/(x^2+y^2)$.
\end{example}
It turns out that $\Delta\psi=0$ at the X-point is sufficient for the existence of well-defined $a$ and $b$
solving Eq.~\eqref{eq:consistency}.
We prove this by actually constructing an algorithm in Section~\ref{sec:construction}.

The following theorem shows what a vanishing Laplacian means in geometrical terms.
Without loss of generality we
assume $\psi(x_X,y_X) = 0$ and call the curve given implicitly
by $\psi(x,y) = 0$ the {\it separatrix}.
\begin{theorem}
  If $\Delta\psi|_{x_X,y_X}=0$ in a given metric $g$,
then the tangent vectors to the separatrix are orthogonal at the X-point in this metric.
\end{theorem}
\begin{proof}
Let us expand $\psi$ around the X-point
\begin{align*}
\psi(x,y) = \frac{1}{2}
(x-x_X, y-y_X)^{\mathrm{T}}
\left . \begin{pmatrix}
\psi_{xx} & \psi_{xy}\\
\psi_{xy} & \psi_{yy}
\end{pmatrix}\right|_{x_X, y_X}
\begin{pmatrix}
x-x_X\\
y-y_X
\end{pmatrix}
+ \dots
\end{align*}
Neglecting higher order terms the equation $\psi(x,y)=0$ then yields a quadratic equation
for $x-x_X, y-y_X$ with the two solution vectors
\begin{align*}
  \vec s_{1,2}
=
\left . \begin{pmatrix}
  -\psi_{xy} \pm \sqrt{\psi_{xy}^2-\psi_{xx}\psi_{yy}}\\
\psi_{xx}
\end{pmatrix}\right|_{x_X, y_X}.
\end{align*}
Now we use that the Laplacian of $\psi$ in the metric $g$ vanishes at the X-point
\begin{align*}
  g^{xx}\psi_{xx} + 2g^{xy} \psi_{xy} + g^{yy} \psi_{yy} = 0,
  \label{}
\end{align*}
where we used that $\psi_x=\psi_y=0$ at the X-point.
With this we can readily compute $ \sum_{ij=1}^2g_{ij} s_1^is_2^j= 0$ that is $\vec s_1$ and $\vec s_2$ are perpendicular at the X-point in the metric $g$.
\end{proof}

Now, we can of course ask what happens if $\Delta\psi\neq0$ at the X-point.
The first observation we make is that for such a non-orthogonal X-point no coordinate system can exist such that $a$ and $b$ as well as their derivatives are bounded.
\begin{theorem} If $a$ and $b$ as well as their derivatives are bounded, then it must hold that $\Delta \psi=0$ at the X-point.
\end{theorem}
\begin{proof}Since $\psi_x = \psi_y = 0$ at the X-point, Eq.~\eqref{eq:consistency} gives $a \Delta \psi =0$. For $a \neq 0$ this can only be satisfied if $\Delta \psi=0$.
\end{proof}

Let us now turn our attention to the case in which either $a$ or $b$ is allowed to diverge at
the X-point. Special cases worth investigating are $a=1$ as in the grid proposed by~\cite{Ribeiro2010} and $b=0$, which yields an orthogonal grid.
It turns out that $\Delta\psi=0$ at the X-point is also a necessary condition for well-defined
$a$ and $b$ to exist in these cases:
\begin{theorem}\label{th:orthogonality} If $\psi$ is a smooth function on a bounded domain that includes a non-orthogonal X-point with $\Delta\psi\neq0$, then there exists no flux-aligned coordinate system $(\zeta,\eta)$ for which $a=1$ and $b$ is well defined at the X-point.
Analogously there exists no flux-aligned coordinate system for which $b=0$ and $a$ is well defined at the X-point.
\end{theorem}
\begin{proof} Substituting $a=1$ into equation (\ref{eq:consistency}) gives
  \[ \{\psi,b\} = - \Delta \psi \]
	and thus
	\[ -\psi_y b_x + \psi_x b_y = - \sqrt{ g}\Delta \psi. \]
	By the method of characteristics we obtain curves $x(t)$, $y(t)$, and $b(t)$ such that
	\[ \dot{x}(t) = -\psi_y(x(t),y(t)), \qquad \dot{y}(t) = \psi_x(x(t),y(t)), \qquad \dot{b} = -\sqrt{g}\Delta \psi(x(t),y(t)). \]
	This implies
	\[ b(t) = b(0) - \int_0^t \sqrt{g}\Delta \psi(x(t'),y(t'))\,\mathrm{d}t'. \]
	Without loss of generality we assume that the X-point is located at $(0,0)$ and that $\Delta \psi(0,0)>0$. Now, we consider characteristics curves $(x_1(t),y_1(t))$ and $(x_2(t),y_2(t))$ such that $\lim_{t\to\infty} (x_1(t),y_1(t)) = 0$ and $\lim_{t\to-\infty} (x_2(t),y_2(t)) = 0$. Then $b(0,0) = \lim_{t\to\infty} b_1(t) = -\infty$ and $b(0,0) = \lim_{t\to-\infty} b_2(t) = \infty$. Thus, we have obtained a contradiction.
The proof for $b=0$ is analogous with the characteristic curves given by $\psi^x(x(t),y(t))$ and $\psi^y(x(t),y(t))$ and we replace $b$ with $\ln a $.
\end{proof}

One might be tempted to conjecture that assuming $a$ is bounded is already enough to rule out the existence of a coordinate system altogether. However, this is not the case as the next example shows.

\begin{example} \label{ex:2}We consider $\psi = \tfrac{1}{2}(2x^2-y^2)$ with $f(\psi)=1$ and $g$ as the canonical metric and look for $\eta$ given by a polynomial. One possibility is
\(\eta = xy\),
which leads to
\[ a = \frac{2x^2 + y^2}{4x^2+y^{2}}, \qquad b=-\frac{xy}{4x^2+y^2}. \]%
    Note that we can easily determine that $a>0$. The volume element is given by
    \( \sqrt{\bar g} = 1/(2x^2 + y^2) \)
    and, as before, diverges as we approach the X-point.
\end{example}

The major difference between the coordinate system considered in Example \ref{ex:2} compared to the orthogonal grid in Example \ref{ex:1} is that the limits of $a$ and $b$ differ as we approach the X-point from different directions.
Thus, there is no uniquely defined value of $a$ and $b$ at the X-point, although a perhaps more serious concern for both coordinate systems is the fact that the volume element $\sqrt{\bar g}$ diverges as we approach the X-point. This is clearly an undesirable property as it means that the X-point is not adequately resolved.

We thus ask the question: is it possible to construct a coordinate system such that the volume element remains bounded as we approach the X-point? The following theorem gives a negative answer.

%

\begin{theorem}\label{th:volume} If $\psi$ is a smooth function on a bounded domain that includes an X-point,
  then there exists no flux-aligned coordinate system $(\zeta,\eta) \in \Omega$, where $\Omega$ is a bounded domain, for which $\sqrt{\bar g}$ is bounded at the X-point.
\end{theorem}
\begin{proof} Substituting equation (\ref{eq:ab}) into $\sqrt{\bar g}^{-1}=(\nabla \psi)^2 a f$ gives
  \[ \sqrt{\bar g}^{-1} = f \{\psi,\eta\} = f( \psi_x \eta_y - \psi_y \eta_x )/\sqrt{g}. \]
    With $\psi$ given and an arbitrary (but fixed) $\sqrt{\bar g}^{-1}$ this yields a first order partial differential equation that can be solved for $\eta$. Employing the method of characteristics we obtain $x(t)$, $y(t)$, and $\eta(t)$ which satisfy the following relations
    \[ \dot{x}(t) = -f \psi_y(x(t),y(t)), \quad
       \dot{y}(t) =  f \psi_x(x(t),y(t)), \quad
       \dot{\eta}(t) = \sqrt{\bar g/g}^{-1}(x(t),y(t)). \]
    Let us assume, without loss of generality, that the X-point is located at $(0,0)$. We now pick a characteristic curve $(x(t),y(t))$ starting at $(x_0,y_0)\neq 0$ that passes through the X-point (the existence of such a curve follows from the fact that at least one coordinate line must pass through the X-point). Since, $(\dot{x}(t),\dot{y}(t))\to 0$ as we approach the X-point, we have $\lim_{t\to\infty} (x(t),y(t)) = 0$ (strictly speaking $t\to-\infty$ is another possibility that is handled by exactly the same argument as is given below for $t\to\infty$).

    Now, let us assume that the volume element $\sqrt{\bar g}$ is bounded as we approach the X-point. Then we can find a constant $\delta$ such that $\sqrt{\bar g}^{-1} \geq \delta$ which implies
    \[ \lim_{t\to\infty} \eta(t) = \lim_{t\to\infty} \int_0^t \sqrt{\bar g/g}^{-1}(x(t),y(t))\,\mathrm{d}t
       \geq \lim_{t\to\infty} \delta t = \infty, \]
    which is a contradiction to the assumption that $(\zeta,\eta)\in\Omega$ with $\Omega$ bounded. Thus, we conclude that $\sqrt{\bar g}\to\infty$ as we approach the X-point.
\end{proof}
In summary, we have proven three major results for structured flux-aligned grids:
\begin{enumerate}
  \item A vanishing Laplacian of the flux-function at the X-point is equivalent to orthogonality of the tangent vectors to the separatrix.
  \item Orthogonal grids and the grid proposed by Reference~\cite{Ribeiro2010} exist if and only if the Laplacian of the flux-function vanishes at the X-point.
  \item The volume element in the transformed coordinate necessarily diverges at the X-point.
\end{enumerate}

\section{Orthogonal grid generation for $\Delta\psi=0$ at the X-point} \label{sec:construction}
In this section we construct an algorithm for the case $\Delta\psi=0$ at the X-point.
We begin to
show how a structured orthogonal grid can be constructed in an arbitrary metric,
then choose a discretization of the computational domain and finally
summarize the proposed algorithm.
Note that this algorithm is an extension of one of our previously suggested algorithms
in Reference~\cite{Wiesenberger2017}.
We consider only two dimensions but let us remark that the extension of the coordinate system to three
dimensions is straightforward in axisymmetric cases\footnote{
Identify $x,y$ with the cylindrical coordinates $R,Z$.
The toroidal angle $\varphi$ is used
as the third coordinate, which is orthogonal to $R,Z$ and
the associated metric element is $g_{\varphi\varphi}=R^2$.
}.

\subsection{Orthogonal grid construction}\label{sec:orthogonal}
In general, the coordinate system $\zeta, \eta$, orthogonal in the prescribed metric $g$, with $\zeta$ aligned to $\psi$, is described by Eq.~\eqref{eq:orthogonal} with $b=0$:
\begin{subequations}
\begin{align}
  \d \zeta & = f(\psi)(\psi_x \d x + \psi_y \d y) \\
  \d \eta  & = a(x,y) \sqrt{g}( -\psi^y \d x + \psi^x \d y)
\end{align}
  \label{eq:orthogonal_monitor}
\end{subequations}
This yields the determinant of the Jacobian matrix
${J}^{-1} = (\nabla\psi)^2 a f\sqrt{g}$.
From the rules of inverse coordinate transformations
we directly see that the contravariant basis vector fields are
\begin{subequations}
\begin{align}
  \partialzeta &= x_\zeta\partialx + y_\zeta\partialy = \frac{1}{(\nabla\psi)^2f} \left(\psi^x \partialx + \psi^y\partialy \right) \\
  \partialeta &=  x_\eta\partialx + y_\eta\partialy =  \frac{1}{(\nabla\psi)^2a\sqrt{g}} \left(-\psi_y \partialx + \psi_x\partialy \right)
\label{eq:orthogonal_linesb}
\end{align}
\label{eq:orthogonal_lines}
\end{subequations}
that is $\partialzeta$ points in the direction of the gradient of $\psi$ and $\partialeta$ in the direction of surfaces given by $\psi=\text{const}$.
We choose $f(\psi)=f_0=\text{const}$.
With our choice we directly get
\begin{align}
  \zeta(x,y) = f_0\psi(x,y)
  \label{eq:zeta}
\end{align}
such that $\zeta=0$ at the separatrix.

As explained in Section~\ref{sec:theory} the function $a(x,y)$ is not arbitrary.  Equation \eqref{eq:consistency} becomes
\begin{align}
  \left(\psi^x\partialx + \psi^y\partialy\right) a = f(\nabla\psi)^2 \partialzeta a= -a\Delta \psi .
  \label{eq:hequation}
\end{align}
In order to integrate Eq.~\eqref{eq:hequation}
we must choose initial conditions for $a$.
This choice together with the normalization of coordinates
depends on the domain that we want to discretize.
Let us remark that if $\Delta\psi=0$ in the whole domain, we directly get a conformal grid with our algorithm. This can be seen as then $a(\zeta,\eta) = f_0$.

\subsection{Domain}
Our goal is to generate a structured grid in a domain
bounded by $\psi_0 < 0$ and $\psi_1 >0$.
We assume that this region forms an ``8'' shape, or a ``surface with two holes'' above
and below the X-point. However, we cut the domain below the X-point.
We are then left with a
region as depicted in Fig.~\ref{fig:topology}.
\begin{figure}[htpb]
    \centering
    \includegraphics[width= 0.8\textwidth]{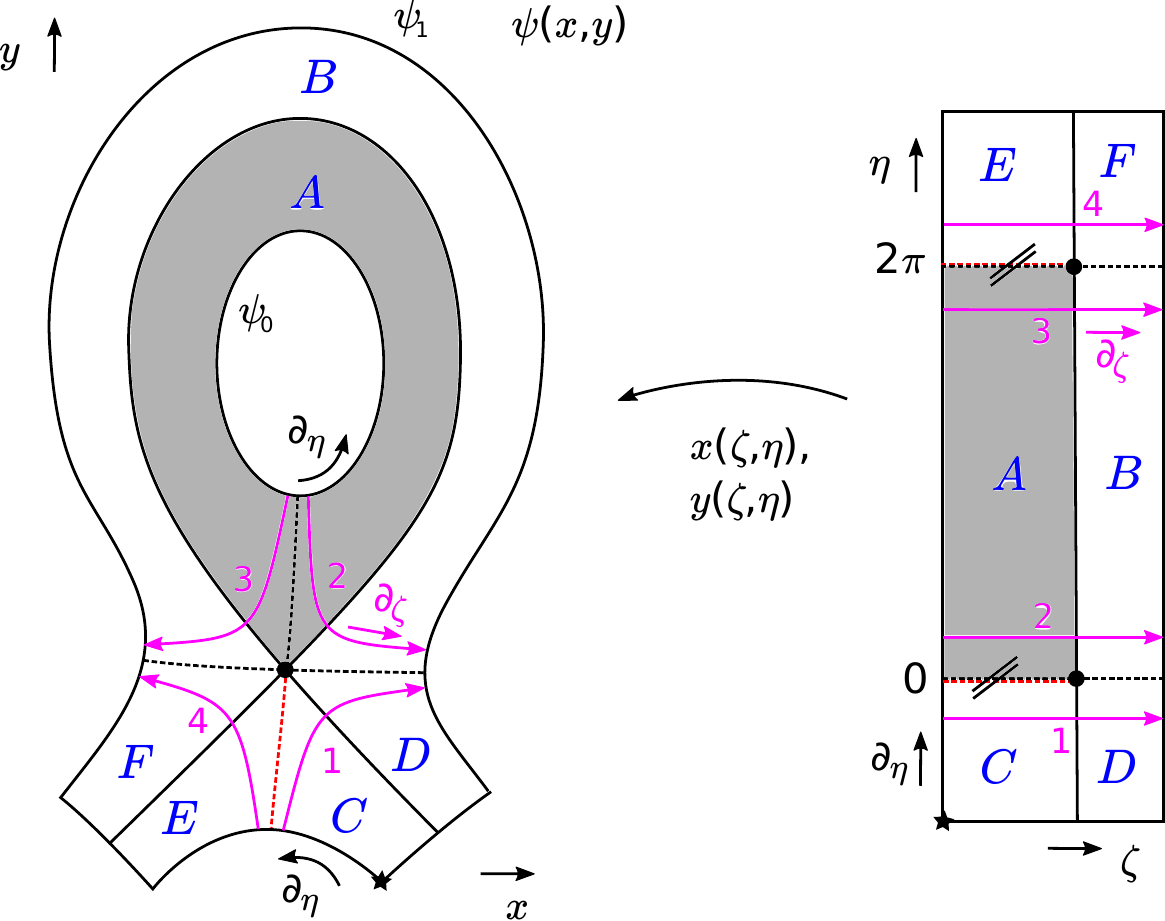}
    \caption{
      Sketch of coordinate transformation. The coordinates $x$ and $y$ (the physical space) are
      pulled back to the orthogonal coordinates $\zeta$, $\eta$ (computational space).
      Note the
      special topology of the 6 coordinate patches induced by the X-point. Patch A is periodic in $\eta$,
      which is depicted by the ``double lines''.
      Patches C and E are connected, which we depict with the red dashed line.
  }
    \label{fig:topology}
\end{figure}
Here, we show a sketch of the
coordinate transformation.
To the left we depict the physical space and to the right the
computational space.
The physical space is covered by 6 coordinate patches labelled A to F.
The topology can be understood by following
the neighbouring coordinate lines 1 and 4 as well as 2 and 3. When passing
the separatrix, line 1 becomes adjacent to line 2, while line 3 changes
neighbour to line 4.
Also note that patch A is periodic and that patches C and E are connected.
This type of structured grid, in which the topology between blocks
has to be separately given, is called {\it block-structured}~\cite{Liseikin}.
When implementing derivatives on the computational grid this
topology has to be taken into account.


\subsection{Normalization}
Now, the problem is how and at what points to fix initial conditions
for $a$ in order to integrate Eq.~\eqref{eq:hequation}. We require
that $a(x,y)$ is continuous and differentiable.
One suggestion would be to set $a=a_0$
on an arbitrary contour line $\psi(x,y) = \psi_0$.
This is indeed a valid choice for the case without X-point.
However, as discussed above, the X-point exchanges the neighbours
of streamlines that pass by it. This means that even though
$a$ is continuous on neighbouring streamlines of $\partial/\partial \zeta$ initially,
it is not guaranteed to be at later stages.
Imagine we chose $a=const$ on the inner flux surface in Fig.~\ref{fig:topology}.
When we integrate Eq.~\eqref{eq:hequation} along $\partial/\partial \zeta$ beyond the separatrix
some streamlines like number $2$ and $3$ change neighbours. This potentially
induces discontinuities among the coordinate patches.
Another uncertainty is the value of $a$ at the
X-point itself since we approach the X-point from two different
sides (line $1$ and $2$ for example).
Similar issues appear if $a$ is chosen on a flux surface on the
outside of the domain.
Our solution to this problem is to initialize $a$ on
the separatrix itself.
Following streamlines along $\partial/\partial \zeta$ away from the separatrix
neighbouring streamlines stay neighbours. Thus, if $a$ is continuous on the separatrix,
the discontinuities among coordinate patches can be avoided.

We trace the separatrix with the streamline of $\partialeta$ at $\zeta=0$.
It can be parameterized by any suitable function.
For the sake of discussion, let us choose the geometric angle $\theta$ defined with respect to an arbitrary point inside the innermost flux surface
(cf. Reference \cite{Wiesenberger2017}).
\begin{subequations}
\begin{align}
  \left . \frac{\d x}{\d \theta}\right|_{\zeta=0}=\frac{x_\eta}{\theta_\eta} &= \frac{-\psi_y}{\psi_x\theta_y - \psi_y \theta_x}\\
  \left . \frac{\d y}{\d \theta}\right|_{\zeta=0}=\frac{y_\eta}{\theta_\eta} &= \frac{\psi_x}{\psi_x\theta_y - \psi_y \theta_x}\\
  \left . \frac{\d \eta}{\d \theta}\right|_{\zeta=0}=\frac{1}{\theta_\eta} &= \frac{\sqrt{g}(\nabla\psi)^2}{\psi_x\theta_y - \psi_y \theta_x}a_0
\end{align}
  \label{eq:etaline}
\end{subequations}
We normalize $a_0$ such that $\eta \in [0,2\pi]$ when we
follow the separatrix in patch A in Fig.~\ref{fig:topology}, that is,
\[ 2\pi = \oint_{\psi=0}\d \eta = \oint_0^{2\pi} \frac{\d \eta}{\d\theta}\bigg|_{\zeta=0}\d \theta \nonumber \]
  or
\begin{align}
  f_0 = a_0 = \frac{2\pi}{\int_0^{2\pi}\d\theta \frac{\sqrt{g}(\nabla\psi)^2}{\psi_x\theta_y - \psi_y \theta_x}}.
  \label{eq:definef}
\end{align}
As initial point for the integration of Eq.~\eqref{eq:etaline} we can use any point with $\psi(x,y) = 0$.
These can be found with a standard bisection
algorithm.
Note that we cannot numerically integrate Eq.~\eqref{eq:etaline}
across the X-point due to the vanishing gradient in $\psi$.
However, we
can integrate towards and close to the X-point. This might be numerically expensive since very small step sizes have to be used, but
can be achieved with sufficient accuracy.

Having chosen values for $f_0$ and $a_0$ the coordinate transformation is now
completely fixed.
In order to find any coordinate $x(\zeta,\eta), y(\zeta, \eta)$ together with its Jacobian
we integrate the streamlines of $\partial/\partial \eta$
and $\partial/\partial \zeta$ given in Eq.~\eqref{eq:orthogonal_lines}.
We do this by first integrating
$\partial / \partial \eta$ on the separatrix (where $a_0$ is known) up to the desired $\eta$ and then following
$\partial / \partial\zeta$ up to $\zeta$
with the obtained starting point.
In order to get $a$ we simply integrate Eq.~\eqref{eq:hequation} along the coordinate
lines.
\begin{subequations}
\begin{align}
    \left .\frac{\d x}{\d \zeta}\right |_{\eta=\text{const}} &= \frac{\psi^x}{f_0(\nabla\psi)^2}\\
    \left .\frac{\d y}{\d \zeta}\right|_{\eta=\text{const}} &= \frac{\psi^y}{f_0(\nabla\psi)^2}\\
    \left .\frac{\d a}{\d \zeta}\right|_{\eta=\text{const}} &= - \frac{\Delta \psi}{f_0(\nabla\psi)^2} a
\end{align}
  \label{eq:etacoordinates}
\end{subequations}
First, however, we need to discuss how the computational domain should be discretized.
\subsection{Discretization of the computational domain} \label{sec:discretization}
As mentioned in the introduction and visible in Fig.~\ref{fig:topology} the computational domain is a product space.
In order to keep this property also numerically we discretize the
$\zeta$ and $\eta$ coordinates separately, i.e. we construct
$N_\zeta$ equidistant cells in $\zeta$
and $N_\eta$ equidistant in $\eta$.
Now, in order to maintain an integer number of equidistant cells in every block
we impose certain restrictions.
Let us define $L^{\text{in}}_\zeta = |f_0\psi_0|$ as the length of blocks A, C and
E in $\zeta$ and $L^{\text{out}}_\zeta = |f_0\psi_1|$ as the length of blocks B, D and F.
With this we have $L_\zeta:=L^{\text{in}}_\zeta + L^{\text{out}}_\zeta$.
Furthermore, we define $L_\eta^\text{out}$ as the length of patch
C, D, E and F in $\eta$ and $L_\eta^\text{in}$ as the length of A and B in $\eta$.
We have $L_\eta:=L_\eta^\text{in} + 2L_\eta^\text{out}$ and now define
\begin{align}
q_\zeta:=L^{\text{out}}_\zeta / L_\zeta\quad
q_\eta := L_\eta^{\text{out}}/L_\eta.
\label{eq:q}
\end{align}
Now, in order to guarantee an integer number of cells in each block
we require that
$N_\zeta^\text{in}  := (1-q_\zeta) N_\zeta$,
$N_\zeta^\text{out} := q_\zeta N_\zeta$,
$N_\eta^\text{in}   := (1-2q_\eta) N_\eta$,
and $N_\eta^\text{out} := q_\eta N_\eta$
are integer numbers.
Note that $q_\zeta$ being rational is a restriction on $L_\zeta$ and thus
on the choice of $\psi_0$ and $\psi_1$. Only one of the two can be chosen freely. Analogously, the condition in $\eta$ is fulfilled by
a proper choice of the boundaries in $\eta$.
Furthermore, with this procedure we, in particular, achieve that the X-point
always appears
as the corner of a cell and never lies inside a grid cell.

\subsection{Grid refinement} \label{sec:refinement}
For the purposes of this study we use a very basic grid refinement
technique. The idea is to simply divide each of $m_\zeta$ cells
in the $\zeta$ coordinate on each side of the separatrix by
$d_\zeta$ equidistant small cells and analogously in $\eta$ we
divide each of $m_\eta$ cells next to the X-point by $d_\eta$.
This means that in total we then have
\begin{align}
N_\zeta^\text{ref} = N_\zeta + 2m_\zeta( d_\zeta -1), \quad
N_\eta^\text{ref} = N_\eta + 4m_\eta (d_\eta-1)
\label{eq:a}
\end{align}
cells in the $\zeta$ and $\eta$ directions.
The factor $2$ in $\zeta$ appears
because we refine the cells left and right
of the separatrix each.
In $\eta$ we have to consider that the X-point appears twice (cf. Fig.~\ref{fig:topology}).
Note that if we were to divide all cells in $\zeta$ by a factor $d_\zeta$, the
refined grid would consist of $N_\zeta^{\text{ref}} = d_\zeta N_\zeta$ cells.
In particular, this means that if the error of a numerical scheme is
dominated by the refined patch, then the refined grid is equivalent to an
unrefined grid with $d_\zeta N_\zeta$ grid points (analogous in $\eta$).
The product space property is preserved in order
to keep the
implementation effort to a minimum.

%

\subsection{Algorithm} \label{sec:algorithm}
Let us finally summarize the grid generation in the following algorithm.
We assume that $\psi_0$ is given and we choose rational numbers
$q_\zeta$ and $q_\eta$ such
that
$\psi_1 = -q_\zeta \psi_0 /(1-q_\zeta)$,
$\eta_0 = -2\pi q_\eta/(1-2q_\eta)$
and
$\eta_1 = 2\pi (1+q_\eta/(1-2q_\eta))$.
Furthermore,
we assume that the $\zeta$ coordinate is discretized by a list of
$N_\zeta$ values $\zeta_i$ with $i = 0,1,\dots N_\zeta-1$
and $\eta$ is discretized by a list of $N_\eta$ values $\eta_j$ with $j = 0,1,\dots N_\eta-1$.
$N_\zeta$ and $N_\eta$ are chosen such that $q_\zeta N_\zeta$ and $q_\eta N_\eta$ are integer
numbers. Finally the list of $\zeta_i$ and $\eta_i$ can be extended by the refinement points
as described in Section~\ref{sec:refinement}.
\begin{enumerate}
  \item Find the X-point. The X-point is often known or can be computed algebraically.
    Numerically, the zeroes of $\nabla \psi$
    can be found very efficiently with a few Newton iterations,
    especially since the Hessian matrix of $\psi$ and its inverse are given analytically.
  \item Find an arbitrary point $(x,y)$ with $\psi(x,y) = 0$
    and a suitable parameterization of Eq.~\eqref{eq:etaline} around
    the X-point.
  \item Integrate Eq.~\eqref{eq:etaline} with $a_0=1$ over $\Theta=[0,2\pi]$ in patch A and use Eq.~\eqref{eq:definef} to compute $f \equiv f_0$ and $a_0 = a(\psi_0)$.
    Use any convenient method for the integration of ordinary differential equations.
  \item Integrate the streamline of Eq.~\eqref{eq:orthogonal_linesb} with $\psi=0$ and $a=f_0$
    from $\eta = 0\dots\eta_j$ for all $j$.
    The result is a list of coordinates $x(0,\eta_j), y(0,\eta_j)$ on the separatrix. These are divided into $N_\eta^{\text{in}}$ coordinates in patch A, and $N_\eta^{\text{out}}$ coordinates in patches C and E each.
  \item Using this list and $a=f_0$ as starting values integrate Eq.~\eqref{eq:etacoordinates}
    from $\zeta=0\dots\zeta_i$ for all $i$ and all $\eta_j$. This gives the map $x(\zeta_i, \eta_j), y(\zeta_i, \eta_j)$ as well as $a(\zeta_i,\eta_j)$ for all $i$ and $j$.
  \item Last, using these results and
    Eq.~\eqref{eq:orthogonal_monitor} evaluate the derivatives
    $\zeta_x(\zeta_i,\eta_j)$, $\zeta_y(\zeta_i, \eta_j)$, $\eta_x(\zeta_i,\eta_j)$, and $\eta_y(\zeta_i, \eta_j)$ for all $i$ and $j$.
\end{enumerate}

\section{The monitor metric approach for $\Delta\psi\neq0$ at the X-point} \label{sec:monitor}
For the following it is important to note
that the theorems in Section~\ref{sec:theory} do not explicitly forbid the existence
of a grid in the case $\Delta\psi\neq0$, only the existence of an orthogonal grid.
The notion of orthogonality, however, and especially the value of $\Delta\psi$, depends
on the given metric tensor $g$ (recall Eq.~\eqref{eq:def_laplacian} at this point).
So what if we were allowed to change the metric tensor $g$
such that $\Delta\psi$ would vanish at the X-point?
For example, consider the canonical metric $g^{ij}=\delta^{ij}$ and
$\psi_{xx}+\psi_{yy}\neq 0$.
If we change the canonical metric to an orthogonal metric
$g^{xx} := -\psi_{yy}$, $g^{yy} := \psi_{xx}$,
we can easily show that $\Delta\psi=0$ at the X-point in this metric.
In this case the consistency equation~\eqref{eq:consistency} allows the existence of
an orthogonal grid.
Indeed, our idea for the construction of a grid for the case $\Delta\psi\neq0$ at
the X-point begins with changing the given metric to a more suitable metric. Then
we use the algorithm in Section~\ref{sec:algorithm} to generate an orthogonal grid in the
changed metric.
This procedure of allowing the metric to be variable instead of a fixed given entity
is called the {\it monitor metric approach}~\cite{Liseikin,Wiesenberger2017}.

Of course, now the question arises what happens to the {\it physical}, Cartesian metric, which we denote $G$ in the following.
So far we have only considered the situation with one metric tensor $g$, the {\it monitor metric}.
The important step is to allow the existence of two metric tensors. The first
one is the artificial monitor metric tensor $g$ and the second one is the
{\it physical} metric tensor $G$.

If we allow two metric tensors in our domain,
 we have in fact two different notions of angles and distances.
 We can measure angles, distances and areas either in $g$ or in $G$.
 This in particular means that if two vectors are orthogonal in one metric they might not
 be in the other.
This is why the monitor metric approach does not violate
our results from Section~\ref{sec:theory}, which are true for both $g$ and $G$.
For example, even if we can construct an orthogonal grid in the monitor metric $g$, in which the Laplacian of $\psi$ vanishes, it is still
non-orthogonal in the physical metric $G$, in which the Laplacian of $\psi$ does not vanish,
and thus does not violate Theorem~\ref{th:orthogonality}, which forbids the existence
of an orthogonal grid for $\Delta\psi\neq 0$.
Unfortunately, there is no way around Theorem~\ref{th:volume} and both volume elements
$\sqrt{\bar g}$ and $\sqrt{\bar G}$
will diverge at the X-point.

It is important to realize that the monitor metric $g$ is an independent
tensor and has nothing to do with the physical metric $G$.
In fact, we would not even need a monitor \textit{metric} tensor.
The formulas in Section~\ref{sec:theory} can be simplified by defining
$\chi^{ij} := \sqrt{g}g^{ij}$ and we could then speak of a {\it monitor tensor} $\chi$, which must be symmetric and positive definite.
This approach would be slightly more general as it
also allows for the inclusion of adaption functions (see Reference~\cite{Wiesenberger2017}).
For this discussion, however,
we keep the metric tensor formulation for the sake of accessibility.

Finally, note that we use the monitor metric $g$ only for the construction of our grid.
The physical equations still use the physical metric tensor $G$, which
therefore also has to be transformed to the new coordinates. This is possible because
with the help of Eq~\eqref{eq:orthogonal_monitor} we numerically construct not only
the grid points $x(\zeta,\eta)$ and $y(\zeta,\eta)$ but also the elements
of the Jacobian matrix. With the Jacobian matrix it is of course  possible to
transform any tensor to the $\zeta,\eta$ coordinate system, in particular the metric tensor $G$.

\subsection{A constant monitor metric}
The task is the construction of a suitable monitor metric.
We suggest the constant tensor
\begin{align}
  g^\text{cte} = \alpha^{-1/2}\left( \frac{\vec v_+ \vec v_+^\mathrm{T}}{\lambda_+} - \frac{\vec v_- \vec v_-^\mathrm{T}}{\lambda_-} \right),
\end{align}
where $\vec v_+$ and $\vec v_-$ are the normalized Eigenvectors of the Hessian matrix of $\psi$ at the X-point.
$\lambda_+$ and $\lambda_-$ are the corresponding Eigenvalues.
Since at the X-point (saddle point) the Hessian matrix is indefinite we can choose $\lambda_-$ to
be the negative and $\lambda_+$ to be the positive Eigenvalue.
We choose $\alpha$ such that the determinant of $g^\text{cte}$ is unity.
With this choice $g^\text{cte}$ is symmetric, positive definite
and $\Delta\psi=0$ at the X-point.
A symbolic calculation shows us the explicit expression
\begin{align}
  g^\text{cte} &= \alpha^{-1/2}\begin{pmatrix}
    \psi_{yy}^2 - \psi_{xx}\psi_{yy}  + 2\psi_{xy}^2 & -(\psi_{xx}+\psi_{yy})\psi_{xy} \\
  -(\psi_{xx} + \psi_{yy})\psi_{xy}                   & \psi_{xx}^2 - \psi_{xx}\psi_{yy} + 2 \psi_{xy}^2
  \end{pmatrix}, \nonumber\\
  \alpha &=\left(\psi_{xy}^2-\psi_{xx}\psi_{yy} \right)\left(\left( \psi_{xx}-\psi_{yy} \right)^2 + 4\psi_{xy}^2\right),
  \label{eq:monitor_const}
\end{align}
where all derivatives of $\psi$ are evaluated at the X-point.
Note that $g^\text{cte}$ reduces to the identity if $\psi_{xx}=-\psi_{yy}$.

\subsection{The bump monitor metric}
As mentioned above, the monitor tensor in Eq.~\eqref{eq:monitor_const} produces non-orthogonal grids. This could be an issue
if orthogonality at the boundary is a requirement, e.g. for the implementation of von Neumann boundary conditions.
In fact, we need the monitor metric to take effect only in the vicinity of
the singularity.
The remaining grid may stay orthogonal in the physical metric $G$.
We therefore introduce the bump-function with amplitude $1$ and radius $\sigma$ centred on the X-point
\begin{align}
\epsilon(x,y) &= \begin{cases}
e^{1 + \left(\frac{(x-x_X)^2}{\sigma^2} + \frac{(y-y_X)^2}{\sigma^2} - 1\right)^{-1}} &\text{ for } (x-x_X)^2+(y-y_X)^2 < \sigma^2 \\
0 &\text{  else}
\end{cases}.
  \label{eq:bump}
\end{align}
With Eq.~\eqref{eq:bump} we introduce
\begin{align}
  g^\text{bump}(x,y) = \bm 1 + \epsilon(x,y) \left( g^\text{cte} - \bm 1\right),
\label{eq:monitor}
\end{align}
where $\bm 1$ is the identity tensor.

\section{Applications of the algorithm} \label{sec:numerics}
In this section we want to test the suggested algorithm in Section~\ref{sec:algorithm}.
We first present a completely analytical scenario and then
proceed
by solving
elliptic equations for a more realistic test case.
Please find codes and implementation details in the latest
\textsc{Feltor} release~\cite{FELTORv4.1}.
Specifically, we generated the results in Sections~\ref{sec:tokamak}-\ref{sec:convergence}
with the programs
\verb!separatrix_orthogonal_t.cu!, \verb!conformalX_elliptic_b.cu!,
\verb!geometryX_refined_elliptic_b.cu! as well as \verb!geometry_diag.cu!
residing in the subdirectory
\verb!feltor/inc/geometries/!.
\subsection{A simple example}
It is instructive to analyse how the algorithm behaves in an analytical example.
To this end let us consider again the flux-function from Example~\ref{ex:2}
\[ \psi = \frac{1}{2}\left( 2x^2 - y^2\right) \] together with the canonical metric tensor.
We directly have $\psi_x = 2x,\ \psi_y = -y$ and $\Delta \psi = 1 \neq 0$ in the canonical metric.
The monitor metric
\eqref{eq:monitor_const} for the present problem becomes
\begin{align*}
  g = \begin{pmatrix}
    1/\sqrt{2} &  0 \\
    0 & \sqrt{2}
  \end{pmatrix}.
\end{align*}
In this monitor metric we have $\psi^x = \sqrt{2}x$ and $\psi^y = \sqrt{2}y$ and $\Delta\psi=0$.
Equation~\eqref{eq:etacoordinates}, parameterized by $y$, is therefore solved by $a=1$ and $x(y) = \pm x_0^2/2y$, which are the contour lines of the ``correct'' $\eta=xy$ coordinate
we discussed in Example~\ref{ex:2}. In Fig.~\ref{fig:examplea} we show the resulting grid.

Now, it is interesting to discuss what goes wrong
if no monitor metric is used in connection with a non-vanishing Laplacian.
Without monitor, Eq.~\eqref{eq:etacoordinates} parameterized by $y$ reads
\begin{align*}
  \frac{\d x}{\d y } = -\frac{2x}{y}, \quad
  \frac{\d y}{\d y } =  1, \quad
  \frac{\d a}{\d y } = \frac{a}{y}.
\end{align*}
As initial conditions for $x$ and $y$ we choose the separatrix given by $y_0 = \pm\sqrt{2}x_0$.
As proposed in the algorithm we choose $a_0=1 =$ const on the separatrix.
We then have the solutions $a(y) = y/y_0$ and $x(y) = \pm  y_0^3 /(y^2\sqrt{2})$,
which we plot in Fig.~\ref{fig:exampleb} and
which lead to
\[ a(x,y) = \left(\frac{|y|}{\sqrt{2} |x|}\right)^{1/3}. \]
This form of $a(x,y)$ is clearly problematic since then $\bar G_{\eta\eta} = 1/(a^2(\nabla\psi)^2)$
and the volume form $\sqrt{\bar G} = 1/(a(\nabla\psi)^2)$
diverge on the $y=0$ line and become $0$ for $x=0$.
Note that $\bar G_{\eta\eta}$ determines the physical cell-size (length)
$l_\eta = \sqrt{\bar G_{\eta\eta}} \triangle \eta$, where $\triangle \eta$ is the cell-size in the computational domain.
If $\bar G_{\eta\eta}$
becomes
very small, then so will $l_\eta$. This is clearly visible in Fig.~\ref{fig:exampleb} on the $x=0$ line.
Now, a large variation in grid-size in a small region of the
physical domain is highly undesirable in any
numerical scheme. In advection type systems a small cell-size deteriorates the CFL condition,
while in inversion problems the large variations in cell-size makes the discretization matrix highly ill-conditioned.
On the other hand, large cell-sizes $\l_\eta$ at $y=0$ mean that this region cannot be accurately resolved.
The cell-sizes seem unproblematic in Fig.~\ref{fig:exampleb} at $y=0$. However, the problem manifests in convergence studies, where the
cell-size $\triangle\eta$ in the computational domain tends to zero.
Since $l_\eta = \sqrt{\bar G_{\eta\eta}}\triangle\eta$ and $\bar G_{\eta\eta}\to \infty$ at $y=0$, the physical cell size $l_\eta$ might not tend to zero
or not at the same rate as $\triangle \eta$. This behaviour
deteriorates or completely inhibits convergence of a numerical scheme.
We therefore conclude that using the algorithm without a proper monitor metric is inadvisable.

\begin{figure}[htbp]
\centering
\subfloat[with monitor (non-orthogonal)]{\includegraphics[trim = 0px 0px 0px 0px, clip, scale=0.6]{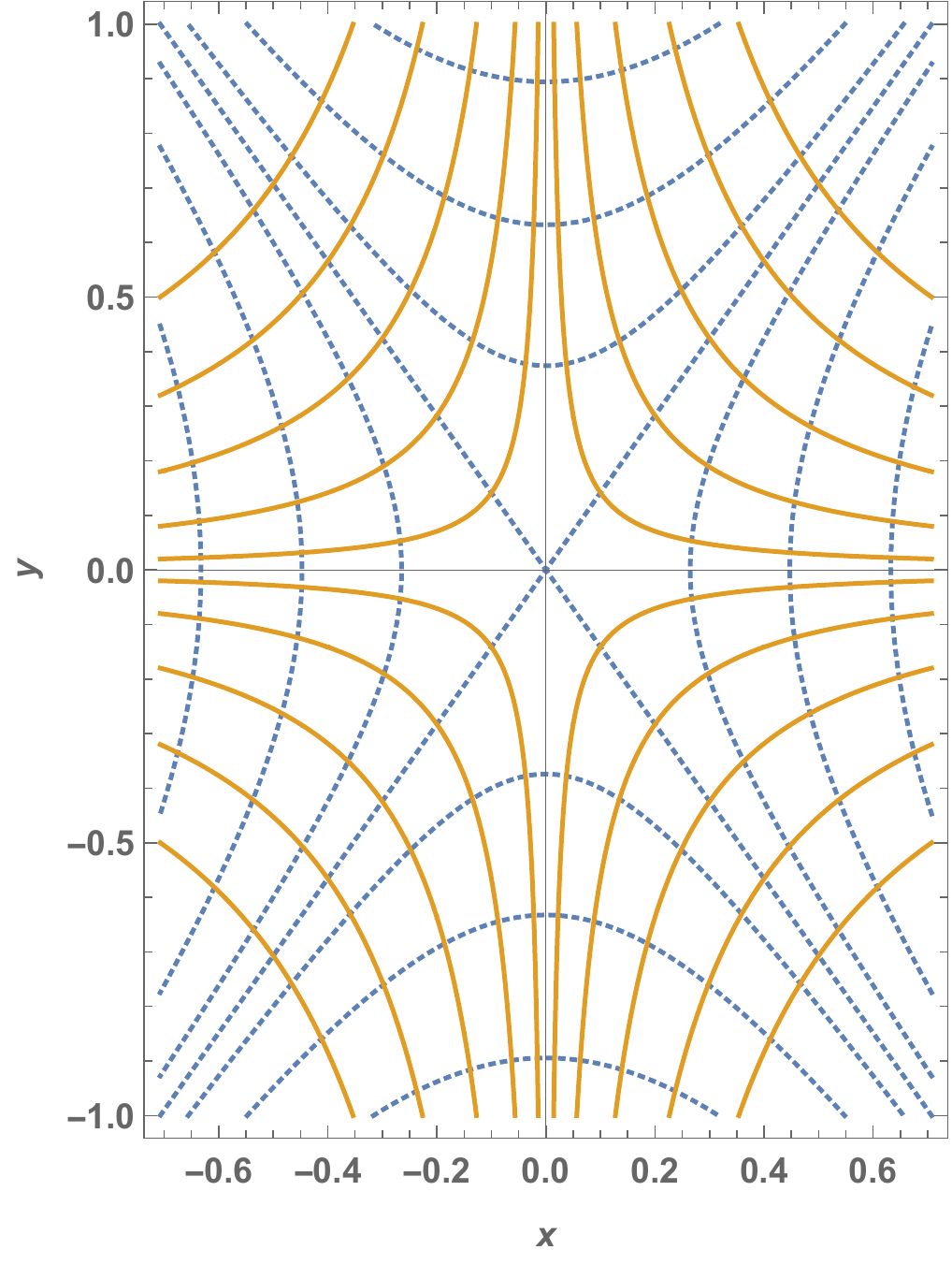}\label{fig:examplea}}
\subfloat[without monitor (orthogonal)]{\includegraphics[trim = 0px 0px 0px 0px, clip, scale=0.6]{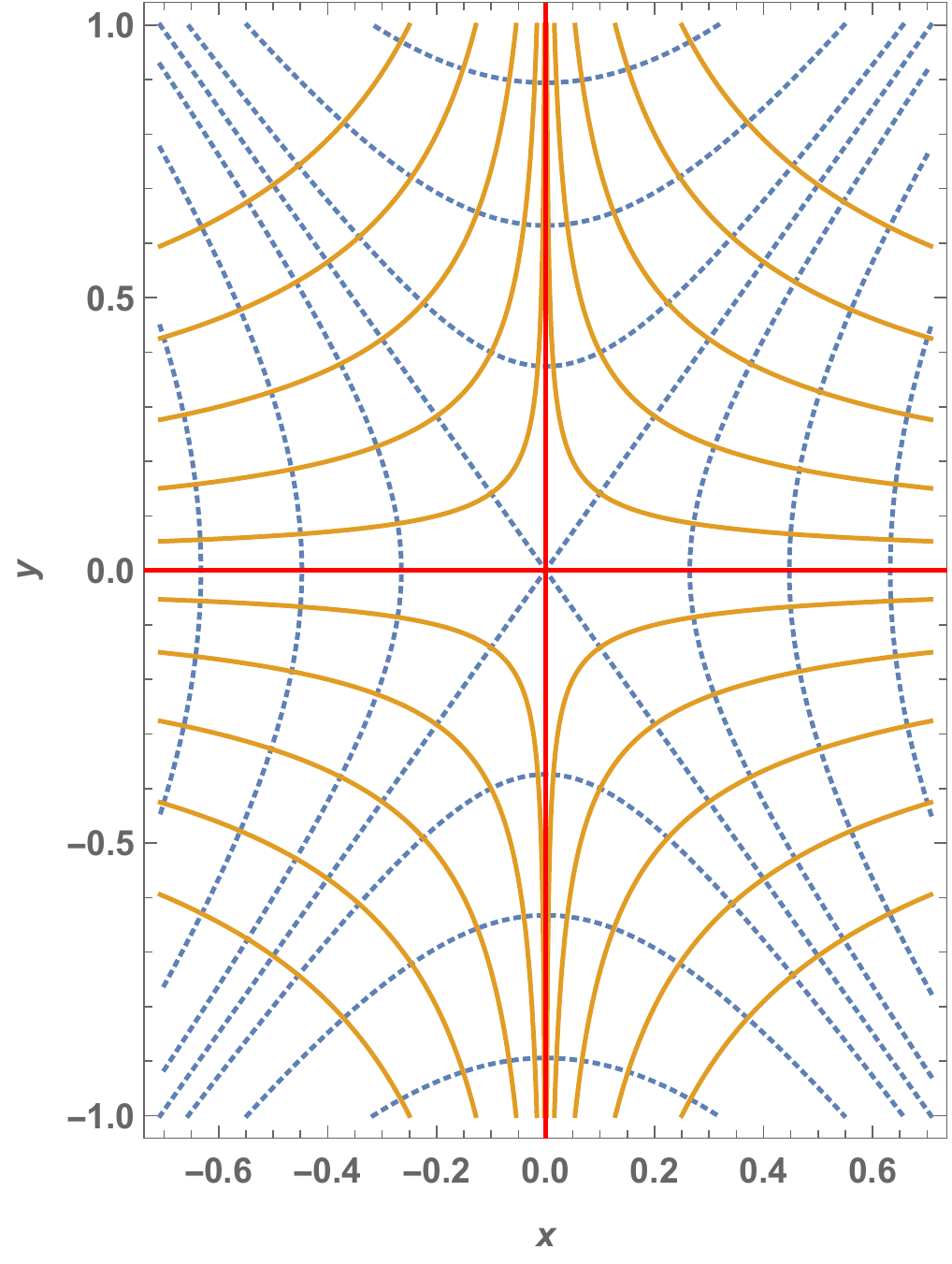}\label{fig:exampleb}}
\caption{
  Grids constructed for $\psi=(2x^2-y^2)/2$ with (a) and without (b) monitor metric.
  We plot the contours $\psi \in \pm\{ 0,0.07, 0.2, 0.4\}$ in dotted blue; in solid orange we have the $\eta$ coordinates through the separatrix $y_0 = \pm \sqrt{2}x_0$ with
  the points $x_0 \in \pm \{ 0.1, 0.2,0.3,0.4,0.5\}$.
  Without monitor metric the physical metric element $\bar G_{\eta\eta}$ becomes $0$ on the $x=0$ line and diverges
  on the $y=0$ line. We indicate this inconsistency with the solid red lines in (b).
}
\label{fig:example}
\end{figure}

\subsection{Tokamak grids}\label{sec:tokamak}
Before we can construct a grid for a realistic scenario we need
to construct an analytical flux-function with X-point.
Reference~\cite{Cerfon2010} presents ``One size fits all'' analytic solutions to
the Grad-Shafranov equation using Solov'ev profiles.
The solution $\psi$ depends on thirteen coefficients. The exact
values reside in the file
\verb!geometry_params_Xpoint.js! in \verb!feltor/inc/geometries! of the accompanying dataset~\cite{FELTORv4.1}.
We will use
this solution for $\psi$ throughout the remainder
of this section.

In Fig.~\ref{fig:gridX} we show the grid produced by our algorithm
with and without refinement.
\begin{figure}[htbp]
\centering
\subfloat[regular grid]{\includegraphics[trim = 0px 0px 0px 0px, clip, scale=0.48]{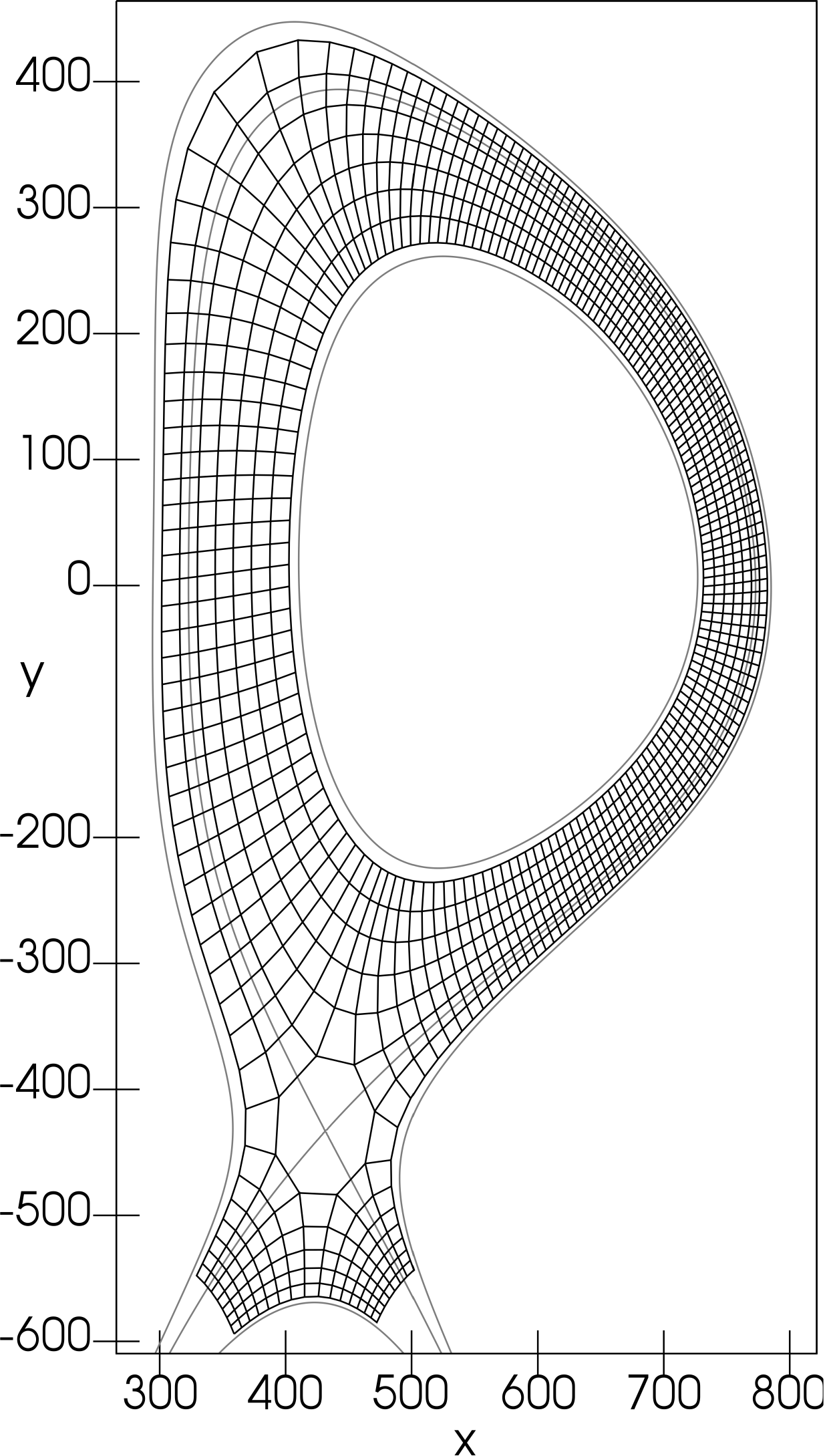}\label{fig:gridXa}}
\subfloat[refined grid]{\includegraphics[trim = 0px 0px 0px 0px, clip, scale=0.48]{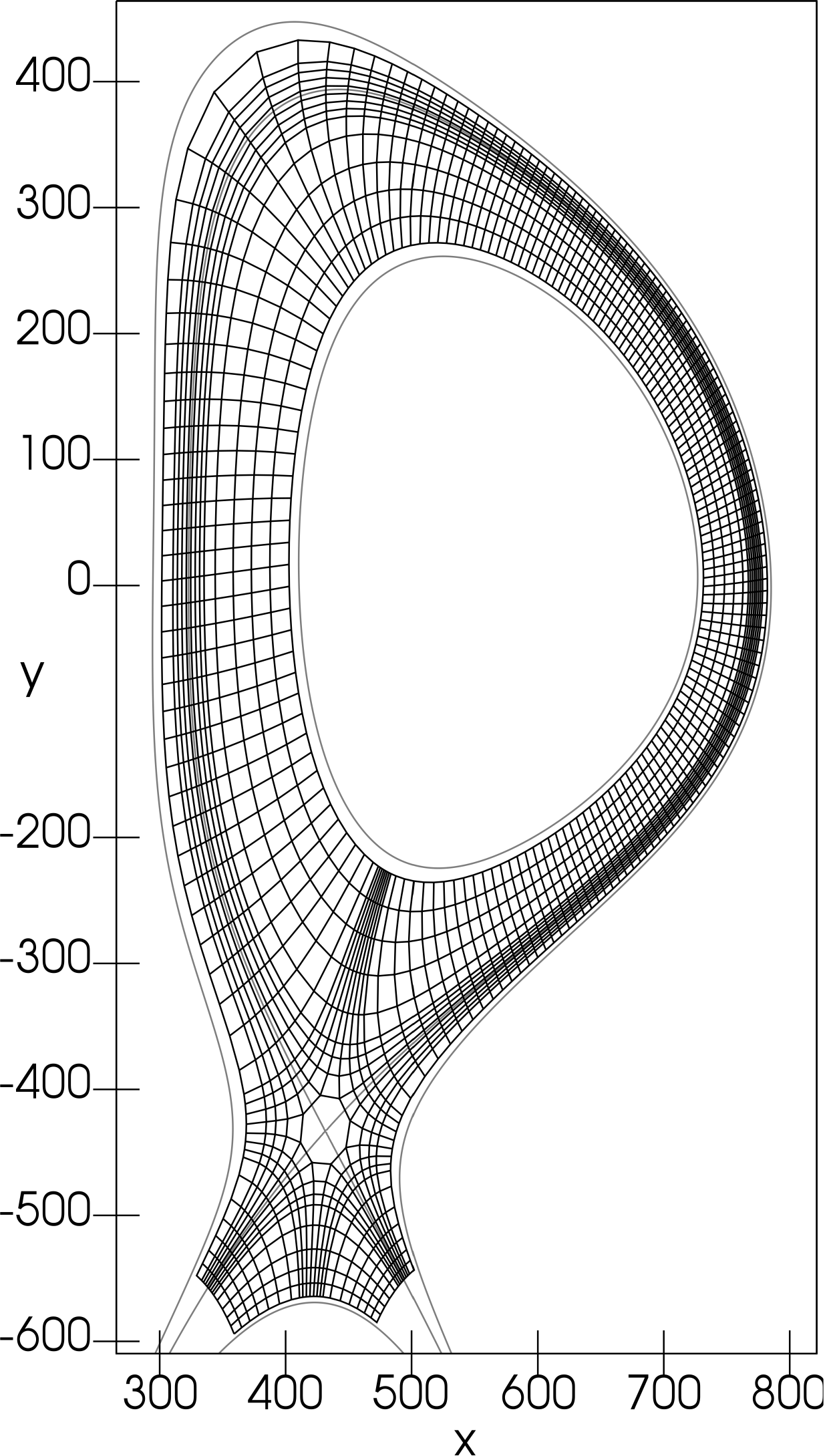}\label{fig:gridXb}}
\caption{
 Orthogonal grids (in the monitor metric Eq.~\eqref{eq:monitor}) with $P=1$, $N_\zeta=8$ and
 $N_\eta = 176$. We have $q_\zeta=1/4$ and $q_\eta=1/22$. Regular grid (a), and
 refined grid with $m_\zeta=m_\eta=1$ and $d_\zeta=d_\eta=4$ (b).
Note that the nodes represent cell centres and not the actual cell boundaries.
The grey lines denote the domain boundaries at $\psi_0=-15$ and $\psi_1=5$.
We also plot the line $\psi=0$ to visualize the separatrix. Note the non-orthogonality
at the X-point.
}
\label{fig:gridX}
\end{figure}
In the regular grid Fig.~\ref{fig:gridXa}
the cell distribution is fairly homogeneous except in the vicinity
of $x=y=350$, where cells become very
large on the outside of the domain, and around the X-point.
In the unrefined grid Fig.~\ref{fig:gridXa} the cells adjacent to the X-point are too large to
sufficiently resolve this area. The resolution is improved in the
refined version of the grid in Fig.~\ref{fig:gridXb}.
Here, we divide the last cell on each side of the X-point
in both the $\zeta$ and the $\eta$ direction by four (i.e. $m_\zeta=m_\eta=1$ and $d_\zeta=d_\eta=4$).

%

The here proposed refinement strategy is sufficient for the
present study.
However, in any production code the refinement techniques should be
re-evaluated.
The downside of the chosen product space refinement is that cells
become unnecessarily small in the regions outside the X-point region.
This is unfavourable for advection type equations.
The goal must be to keep the cell sizes as
homogeneous as possible across the domain so as not to deteriorate
the CFL condition for advection-diffusion type problems.
Fortunately, the X-point is a single point such that
the refinement is local and shouldn't present any performance issues.
A direct solution could be giving up the product space property
of the computational space and restricting the refinement to the area
around the X-point. This, however, increases the implementation complexity.
Let us point out here that there are advanced techniques available
that might be worth considering for an efficient implementation.
For elliptic grids it is known that with the help of adaption
functions and monitor metrics the distribution of cells across the
domain can be controlled. In this way the coordinate transformation
itself includes a grid
refinement~\cite{Liseikin, Glasser2006, Wiesenberger2017, Vaseva2009}.
Although these techniques
are very powerful their applicability to the present case
remains to be explored. The difficulty lies in the fact that an elliptic
equation has to be inverted on the domain,
which can prove difficult to achieve due to the diverging metric at the X-point.
A converging solver, however, is a prerequisite for the generation of elliptic grids. This motivates the following study.

\subsection{Discretization of an elliptic equation}
The two-dimensional elliptic equation
  $\Delta \phi  = \rho$
  in the new coordinates reads
\begin{align}
  \frac{\partial}{ \partial\zeta}\left(\sqrt{\bar G} \left(\bar G^{\zeta\zeta}
  \frac{\partial\phi}{\partial\zeta} + \bar G^{\zeta\eta} \frac{\partial\phi}{\partial\eta} \right)\right) + 
  \frac{\partial}{\partial\eta}\left(\sqrt{\bar G} \left(\bar G^{\eta\zeta}
  \frac{\partial\phi}{\partial\zeta} + \bar G^{\eta\eta} \frac{\partial\phi}{\partial\eta} \right)\right)
  &=\sqrt{\bar G}\rho,
    \label{eq:elliptic}
\end{align}
where we multiplied with the
volume element to make the left hand side symmetric.
In this equation $\bar G$ denotes the Cartesian metric $G$ transformed to the
new coordinate system.
We use this equation to test the quality of our grids.

We use a local discontinuous Galerkin method to discretize this equation
on the computational ($\zeta,\eta$) domain~\cite{Cockburn2001}.
This method approximates the solution by a order $P-1$ polynomial in each
cell with $P$ being the number of polynomial coefficients.
In contrast to finite element methods the approximation is allowed to
be discontinuous at cell boundaries.
As described in Reference~\cite{Held2016} we compute the left side
of Eq.~\eqref{eq:elliptic} by discretizing the first derivatives
$\partial/\partial\zeta$ and $\partial/\partial\eta$ with a forward discretization.
These are just the
discretizations we would have for the discretization of first
derivatives in a Cartesian grid. Of course, we need to
take into account the special topology of the computational space (cf.~Fig.~\ref{fig:topology}).
Note that the derivative is a topological entity, which means
that no metric is needed to define a directional derivative on a
manifold~\cite{Frankel}.
The metric elements can be multiplied
to the first derivatives by simple point-by-point multiplication.
The second derivatives can be computed by using the adjoint of the
first derivatives. Note that in the local discontinuous Galerkin scheme
we need to add jump terms to the discretizations to penalize the discontinuities at the
cell boundaries~\cite{Cockburn2001}. Without these the
numerical solutions fail to converge at all.
We are then finally left with a self-adjoint
discretization of the elliptic operator.

It is a priori unclear whether our numerical scheme can cope with the diverging metric elements at the X-point, even if the metric or any other function is never evaluated
at the X-point itself.
Note that the coordinate singularity is weak
in the sense that the integration over the volume element
yields the correct volume of the domain.
We verified this numerically, that is we numerically evaluate the volume $\int\d\zeta\d\eta \sqrt{\bar G(\zeta,\eta) }$ using Gauss--Legendre integration in computational space. We find equivalent results to integrating $\int_\Omega \d x\d y$ directly (with $\Omega$ being the physical domain) in
Cartesian coordinates using a simple quadrature rule.
Thus, the weakly formulated discontinuous Galerkin
scheme should be able to cope with the diverging metrics without any necessary adaptions.

\subsection{Convergence tests}\label{sec:convergence}
We now test the convergence with
a ``bump'' solution
\begin{align}
\phi(x,y) &= \begin{cases}
e^{1 + \left(\frac{(x-x_0)^2}{\sigma^2} + \frac{(y-y_0)^2}{\sigma^2} - 1\right)^{-1}} &\text{ for } (x-x_0)^2+(y-y_0)^2 < \sigma^2 \\
0 &\text{  else}
\end{cases}
  \label{eq:sol_blob}
\end{align}
with centre $(x_0,y_0)=(480,-300)$ and
radius $\sigma=70$. This solution has no variation across the X-point situated at
approximately $(x_X,y_X)=(431,-433)$. The boundary conditions are homogeneous Dirichlet in
$\zeta$ and $\eta$.
We plot the analytic solution Eq.~\eqref{eq:sol_blob} for the given parameters
in Fig.~\ref{fig:two_blobs}.
\begin{figure}[htpb]
    \centering
    \includegraphics[width= 0.6\textwidth]{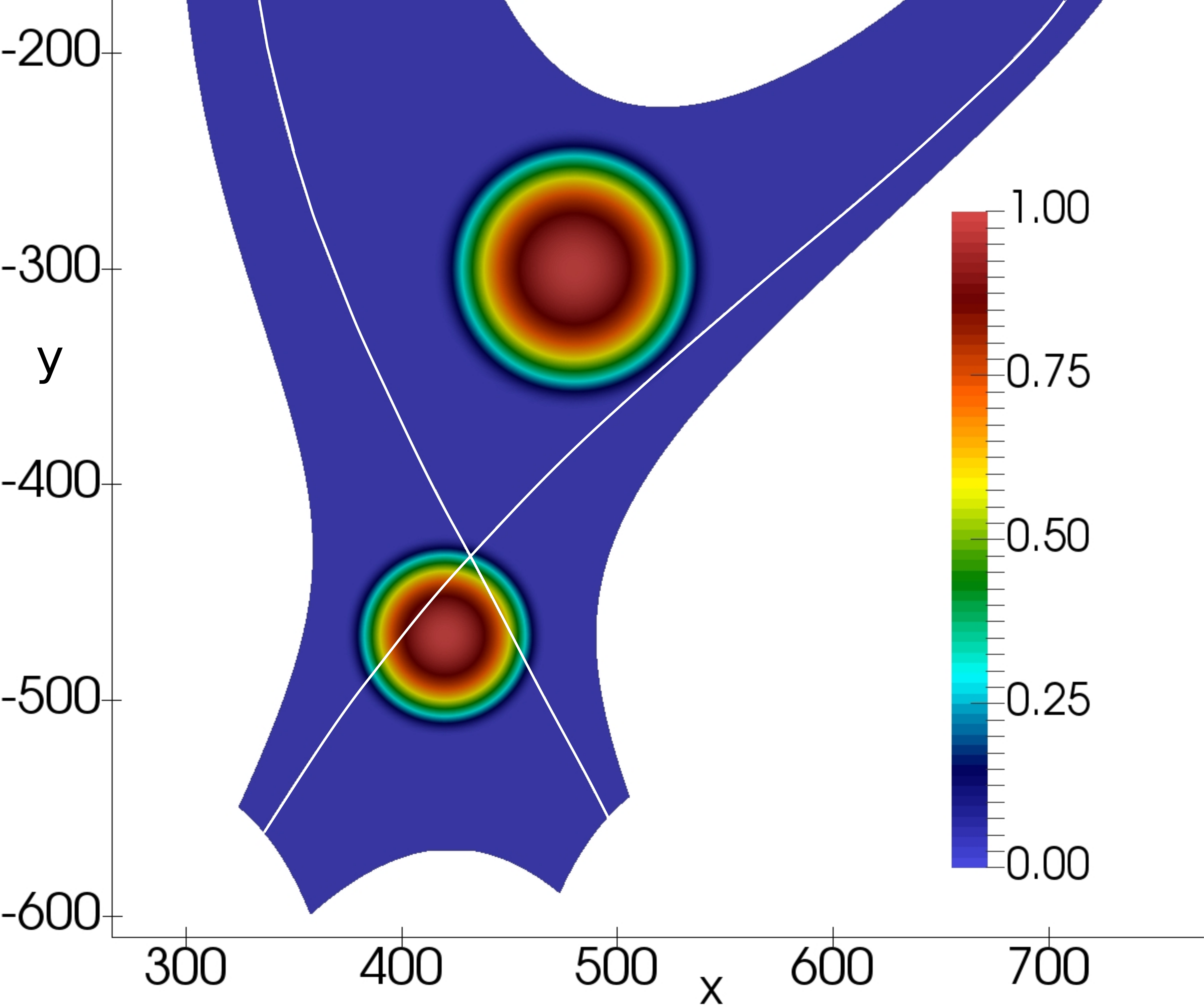}
    \caption{
      The location of our two ``bump'' solutions,
      the upper located at $(x_X,y_X)=(431,-433)$
      with $\sigma=70$ and the lower bump located near the X-point at $(x_0,y_0)=(420,-470)$ with $\sigma=50$.
      We also indicate the location of the separatrix with the white line.
  }
    \label{fig:two_blobs}
\end{figure}
We can insert Eq.~\eqref{eq:sol_blob} into Eq.~\eqref{eq:elliptic} to compute the
corresponding right hand side analytically. With this right hand side given we then compute
a numerical solution $\phi_\text{num}$ to Eq.~\eqref{eq:elliptic}.
The relative error and the order of convergence can be defined in the $L_2$ norm as
\begin{align}
  \eps = \left(\frac{\int_{\zeta_0}^{\zeta_1}\d \zeta \int_{\eta_0}^{\eta_1}\d \eta \sqrt{\bar G} (\phi_{\text{num}} - \phi_{\text{ana}})^2 }
  {\int_{\zeta_0}^{\zeta_1}\d \zeta \int_{\eta_0}^{\eta_1}\d \eta \sqrt{\bar G} \phi_{\text{ana}}^2 }\right)^{1/2}, \quad \mathcal O = \frac{\ln[ \eps(2N)/\eps(N)]}{\ln(2)},
  \label{eq:error_norm}
\end{align}
where $\sqrt{\bar G}\d \zeta \d \eta$ is the correct volume form in the $\zeta,\eta$ coordinate system.
The order $\mathcal O$
is computed via two consecutive errors, between which the number of cells $N$ is doubled.
\begin{table}[htbp]
\begin{tabular}{cccccccccc}
\toprule
\multicolumn{2}{c}{}& \multicolumn{2}{r}{P=1} &  \multicolumn{2}{r}{P=2}  & \multicolumn{2}{r}{P=3}  & \multicolumn{2}{r}{P=4}  \\
$N_\zeta$  & $N_\eta$ & error & order & error & order& error & order& error & order  \\
\midrule
4 & 88 & 5.46E+00 & \multicolumn{1}{l}{} & 5.71E-01 & \multicolumn{1}{l}{} & 5.95E-01 & \multicolumn{1}{l}{} & 1.81E-01 & \multicolumn{1}{l}{} \\
8 & 176 & 4.53E-01 & 3.59 & 4.37E-01 & 0.39 & 1.35E-01 & 2.14 & 3.01E-02 & 2.59 \\
16 & 352 & 2.33E-01 & 0.96 & 4.68E-02 & 3.22 & 1.06E-02 & 3.68 & 2.08E-03 & 3.86 \\
32 & 704 & 1.99E-02 & 3.55 & 6.28E-03 & 2.90 & 8.14E-04 & 3.70 & 1.59E-04 & 3.71 \\
64 & 1408 & 8.63E-03 & 1.21 & 1.33E-03 & 2.24 & 9.66E-05 & 3.07 & 1.32E-05 & 3.60 \\
128 & 2816 & 4.41E-03 & 0.97 & 3.34E-04 & 2.00 & 1.32E-05 & 2.87 & 1.70E-06 & 2.95 \\
\bottomrule
\end{tabular}
\caption{
Convergence table for the bump solution Eq.~\eqref{eq:sol_blob} away from the X-point
for various polynomial orders. No grid refinement is used.
    Error and order are defined via Eq.~\eqref{eq:error_norm}.
    Average orders are:
    ($2.06$, $P=1$);
    ($2.15$, $P=2$);
    ($3.09$, $P=3$);
    ($3.34$, $P=4$).
}
\label{tab:blob}
\end{table}


In Table~\ref{tab:blob} we show the error and corresponding
orders for various polynomial orders and grid resolutions.
A ratio of $N_\eta/N_\zeta=20$ is chosen such that the aspect
ratio of the resulting cells is approximately unity.
We observe a rather irregular convergence for all values of $P$.
We attribute this behaviour to the irregular shapes of the grid cells at
the location of the bump comparing Fig.~\ref{fig:two_blobs} to Fig.~\ref{fig:gridX}.
Although in this example there is no variation of the solution at the X-point
the grid cells nevertheless
become larger in its vicinity. Also the aspect ratio of the cells in the upper half
of the bump are different from the aspect ratio in the lower half.
We compute an average order of convergence with the values at $N_\zeta=4$ and
$N_\zeta=128$ in an attempt to smooth the variations. The expected order $P$
is then approximately recovered for $P=2$ and $P=3$. For $P=1$ the average order is approximately
$2$, however, the orders of the two finest grids indicate that only the expected first order is
recovered. For $P=4$ the computed average order is more than $20$\% too small. On the other
side the absolute errors in the $P=4$ grids are the smallest among all grids.
Finally, note that we also observed this irregular convergence in a similar example in Reference~\cite{Wiesenberger2017}, where no X-point was present in the domain.

Let us now turn our attention to the case when variations around the X-point appear in the
solution. We use the bump defined in Eq.~\eqref{eq:sol_blob}
with $(x_0,y_0)=(420,-470)$ and $\sigma=50$. This is the lower bump in Fig.~\ref{fig:two_blobs}.
\begin{table}[htbp]
\begin{tabular}{cccccccccc}

\toprule
\multicolumn{2}{c}{}& \multicolumn{2}{r}{P=1} &  \multicolumn{2}{r}{P=2}  & \multicolumn{2}{r}{P=3}  & \multicolumn{2}{r}{P=4}  \\
$N_\zeta$  & $N_\eta$ & error & order & error & order& error & order& error & order  \\
\midrule
4 & 88 & 1.65E+01 & \multicolumn{1}{l}{} & 1.36E+00 & \multicolumn{1}{l}{} & 7.63E-01 & \multicolumn{1}{l}{} & 8.74E-01 & \multicolumn{1}{l}{} \\
8 & 176 & 2.71E+00 & 2.60 & 6.48E-01 & 1.07 & 7.32E-01 & 0.06 & 3.28E-01 & 1.42 \\
16 & 352 & 3.74E-01 & 2.86 & 8.43E-01 & -0.38 & 3.70E-01 & 0.98 & 1.48E-01 & 1.15 \\
32 & 704 & 5.72E-01 & -0.61 & 2.98E-01 & 1.50 & 8.55E-02 & 2.11 & 8.84E-02 & 0.74 \\
64 & 1408 & 1.95E-01 & 1.55 & 6.25E-02 & 2.25 & 1.83E-02 & 2.22 & 2.95E-02 & 1.58 \\
128 & 2816 & 5.36E-02 & 1.86 & 3.48E-02 & 0.85 & 4.14E-02 & -1.18 & 4.89E-03 & 2.59 \\
\bottomrule
\end{tabular}
\caption{
Convergence table for the bump solution Eq.~\eqref{eq:sol_blob} on the X-point
for various polynomial orders. No grid refinement is used.
    Error and order are defined via Eq.~\eqref{eq:error_norm}.
    Average orders are:
    ($1.65$, $P=1$);
    ($1.06$, $P=2$);
    ($0.84$, $P=3$);
    ($1.50$, $P=4$).
}
\label{tab:blobX}
\end{table}

In Table~\ref{tab:blobX} we show the results of the same experiment as in Table~\ref{tab:blob}. Also in this case the convergence rates are highly irregular and even negative in some
cases. Again, we compute the average orders, which this time lie between $0.86$ and $1.5$.
From this we conclude
that the X-point reduces the convergence rate to
around order $1$ for all numbers of polynomial coefficients.
Inspection of the error reveals that indeed the error is entirely
dominated by the region around the X-point.
We attribute the loss of convergence
to the diverging metric elements.
As seen in Fig~\ref{fig:gridX} these lead to large cell sizes in $\zeta$ and $\eta$.
If we define the cell size in the computational domain as
$\Delta_\zeta:=L_\zeta/N_\zeta$ and $\Delta_\eta:=L_\eta/N_\eta$,
we compute the cell sizes in the physical domain by
\begin{subequations}
\begin{align}
  l_\zeta &:= \sqrt{\bar G_{\zeta\zeta}} \triangle\zeta = \sqrt{\bar G}\sqrt{ \bar G^{\eta\eta}} \triangle\zeta  \propto \sqrt{\bar g}\sqrt{ \bar g^{\eta\eta}} \triangle\zeta =\left(a_0|\nabla\psi|\right)^{-1}\triangle\zeta,\\
  l_\eta &:= \sqrt{\bar G_{\eta\eta}} \triangle\eta = \sqrt{\bar G}\sqrt{ \bar G^{\zeta\zeta}} \triangle\eta \propto \sqrt{\bar g}\sqrt{ \bar g^{\zeta\zeta}} \triangle\eta = \left(f_0|\nabla\psi|\right)^{-1} \triangle\eta,
\end{align}
\end{subequations}
where we approximate the length in the actual Cartesian metric $G$ with the length in the monitor metric $g$.
Clearly, the cell sizes $l_\zeta$ and $l_\eta$ diverge at the
X-point due to the vanishing gradient in $\psi$.
Now in the previous tests we looked for convergence in terms of $\triangle\zeta$ and $\triangle\eta$
that is  $\varepsilon \propto \triangle\zeta^P$.
However, it could be argued that the error should be proportional to
$l_\zeta$ instead of $\triangle\zeta$. As long as $|\nabla\psi|$ is
well-behaved the definitions are the same, but at the X-point
this makes a difference.
This means that even though we reduce $\triangle\zeta$ and $\triangle\eta$
in the computational domain, $l_\zeta$ and $l_\eta$ do not shrink with the same rate, which
might explain the reduced orders in Table~\ref{tab:blobX}.

In order to remedy the loss of convergence due to large $l_\zeta$ and $l_\eta$ we
use the grid refinement from Section~\ref{sec:refinement}. The grid refinement has the goal to reduce the
sizes $\triangle\zeta$ and $\triangle\eta$ locally around the X-point until the physical
lengths $l_\zeta$ and $l_\eta$ at the X-point equal the lengths in the remaining regions of the grid.
If the error at the X-point is small enough, the error should then be dominated by the
error in the remaining grid.
Theoretically, this should then restore the expected order.

Numerically, we test this hypothesis using again the bump on the X-point as a solution to Eq.~\eqref{eq:elliptic}.
\begin{table}[htbp]
\begin{tabular}{lllllllll}
\toprule
\multicolumn{1}{l}{$N_\zeta\times N_\eta$} & \multicolumn{2}{l}{$4\times 88$}  & \multicolumn{2}{l}{$8\times 176$}   & \multicolumn{2}{l}{$16\times 352$ }  & \multicolumn{2}{l}{$32\times 704$}   \\
$d_\zeta=d_\eta$  & error & order & error & order & error & order & error & order  \\
\midrule
1 & 7.63E-01 &  & 7.32E-01 &  & 3.70E-01 &  & 8.55E-02 &  \\
2 & 7.32E-01 & x & 3.23E-01 &  & 9.54E-02 &  & 1.81E-02 &  \\
4 & 3.70E-01 &  & 1.36E-01 & 2.43  & 2.19E-02 & 2.64 & 4.06E-02 &  \\
8 & 8.54E-02 &  & 1.54E-01 &  & 5.28E-02 &  & 6.96E-03 &  \\
16 & 1.79E-02 &  & 1.31E-01 &  & 1.41E-02 & & 2.66E-03 & 3.04 \\
32 & 5.38E-02 &  & 1.51E-01 &  & 1.68E-02 &  & 2.18E-03 &  \\
\bottomrule
\end{tabular}
\caption{
Convergence table for the bump solution Eq.~\eqref{eq:sol_blob} on the X-point
for $P=3$ and increasing refinement. Note that for $N_\zeta\times N_\eta = 4\times 88$ the bump lies entirely in the refined region.
    Error and order are defined via Eq.~\eqref{eq:error_norm}. The orders are
    computed with the error values in the same row the orders are indicated.
    For $N_\zeta\times N_\eta> 4\times 88$ these are the rows where the error starts to stagnate with increasing refinement.
}
\label{tab:blobX_refined}
\end{table}

In Table~\ref{tab:blobX_refined} we show results for a fixed value $P=3$.
We start with unrefined grids ($d_\zeta=d_\eta=1$) of increasing resolutions $N_\zeta$ and $N_\eta$.
Then, we divide the last cells adjacent to the X-point into $d_\zeta=d_\eta$ parts and
repeat the inversion of Eq.~\eqref{eq:elliptic}.
Note that since we only refine the last cells adjacent to the X-point the actual region
in the physical domain
that is refined changes with grid resolutions.
This leads to the effect that the solution for the lowest resolution lies
entirely in the refined region and thus grid-refinement always leads to an improved
error. In this case, the errors are equal to the corresponding $P=3$ column in Table~\ref{tab:blobX}, because, as discussed in Section~\ref{sec:refinement}, the
refined grids are equivalent to the unrefined grids of increased resolution.
Only for higher resolutions in $N_\zeta\times N_\eta$ we observe error stagnation for higher grid refinement.
If we compute the order with the stagnating values, we
recover $P\approx 3$.

In order to be entirely certain that the error is dominated by the unrefined region
we repeat our experiment with
the sum of both upper and lower bumps as a solution to Eq.~\eqref{eq:elliptic} visible in Fig.~\ref{fig:two_blobs}.
\begin{table}[htbp]
\begin{tabular}{lllllllll}
\toprule
\multicolumn{1}{l}{$N_\zeta\times N_\eta$} & \multicolumn{2}{l}{$4\times 88$}  & \multicolumn{2}{l}{$8\times 176$}   & \multicolumn{2}{l}{$16\times 352$ }  & \multicolumn{2}{l}{$32\times 704$}   \\
$d_\zeta=d_\eta$  & error & order & error & order & error & order & error & order  \\
\midrule
1 & 7.16E-01 &  & 4.57E-01 &  & 2.18E-01 &  & 4.96E-02 &  \\
2 & 5.81E-01 &  & 1.92E-01 &  & 5.84E-02 &  & 1.06E-02 &  \\
4 & 3.72E-01 & x & 1.21E-01 & 1.62 & 1.55E-02 & 2.97 & 2.35E-02 &  \\
8 & 2.71E-01 &  & 1.33E-01 &  & 3.43E-02 &  & 4.19E-03 &  \\
16 & 2.55E-01 &  & 1.18E-01 &  & 1.23E-02 &  & 1.66E-03 & 3.23 \\
32 & 3.29E-01 &  & 1.33E-01 &  & 1.41E-02 &  & 1.43E-03 &   \\
\bottomrule
\end{tabular}
\caption{
  Convergence table for the two bumps solution shown in Fig.~\ref{fig:two_blobs}
for $P=3$ and increasing refinement.
    Error and order are defined via Eq.~\eqref{eq:error_norm}. The orders are
    computed with the error values in the same row the orders are indicated.
    These are the rows where the error starts to stagnate with increasing refinement.
}
\label{tab:two_blobs}
\end{table}

In Table~\ref{tab:two_blobs} we show results for a fixed value $P=3$.

Now, the error first decreases and then stagnates even for $N_\zeta\times N_\eta=4\times88$. The stagnating values
are comparable to the stagnating values in Table~\ref{tab:blobX_refined} and the values in the $P=3$ column of Table~\ref{tab:blob}.
The latter observation strongly supports the conclusion
that with enough refinement at the X-point the error is dominated by
the error in the unrefined region.
If we compute the order with the stagnating values, we
indeed recover $P\approx 3$. The first value of $1.62$ at the $8\times 176$ could be
explained by the relatively large errors in the $4\times 88$ and $8\times 176$ grids.
Convergence only sets in at higher resolutions.

\section{Conclusion}
In summary we make two statements.
First, a structured aligned orthogonal grid
can be consistently constructed only when the separatrix
forms a right angle ($\Delta\psi=0$) at the X-point. We discuss how with the help of a monitor metric
the notion of orthogonality can change in a way that a grid construction is
possible.
This is based on our theoretical analysis
and the following discussion of our algorithm for structured grid generation.
Second, convergence of a numerical discretization of an elliptic equation on the grid may reduce to order one due to the diverging volume element
or cell sizes at
the X-point. Our local discontinuous Galerkin discretization
converges with order approximately $P$ only
as long as the solution is constant around the X-point.
We show that grid refinement is needed around the
X-point in order to achieve convergence at order greater than one, if the solution varies across the X-point. This is the typical situation in a practical application of the grid.


\section*{Acknowledgements}
The research leading to these results has received funding from the European Union's Horizon 2020 research and innovation programme under the Marie Sklodowska-Curie grant agreement no. 713683 (COFUNDfellowsDTU).
This work was supported by the Austrian Science Fund (FWF) Y398.

\bibliography{refs}
\bibliographystyle{iopart-num.bst}


\end{document}